\documentclass[review]{elsarticle}
\usepackage{hyperref}

\usepackage[utf8]{inputenc}
\usepackage[english]{babel}

\usepackage{graphicx}

\usepackage{amsmath,amssymb,amsthm}

\usepackage{enumerate}

\usepackage{xcolor}
\usepackage{todonotes}

\DeclareGraphicsExtensions{.pdf,.png,.jpg,.eps}

\def\hmath$#1${\texorpdfstring{{\rmfamily\textit{#1}}}{#1}}

\def\begsubequ{\begin{subequations}}
\def\endsubequ{\end{subequations}}
\newtheorem{lemma}{Lemma}
\newtheorem{proposition}{Proposition}
\newtheorem{corollary}{Corollary}
\newtheorem{fact}{Fact}
\newtheorem{definition}{Definition}
\newtheorem{remark}{Remark}

\newtheorem{assumption}{Assumption}

\def\bul{\noindent $\bullet\;\;$}

\def\begcen{\begin{center}}
\def\endcen{\end{center}}

\newcommand{\col}{ \mbox{col} }

\def\caly{{\cal Y}}

\def\calo{{\cal O}}
\def\calh{{\cal H}}

\def\calh{{\cal H}}

\def\call{{\cal L}}

\def\hal{{1 \over 2}}

\def\litcallinf{\ell_\infty}
\def\intnum{\mathbb{Z}}

\def\L2{{\cal L}_2}
\def\L2e{{\cal L}_{2e}}
\def\callinf{{\cal L}_\infty}
\def\bul{\noindent $\bullet\;\;$}
\def\rea{\mathbb{R}}

\def\adj{\mbox{adj}}

\def\begmat#1{\begin{bmatrix}#1\end{bmatrix}}
\def\begali#1{\begin{align}{#1}\end{align}}
\def\begalis#1{\begin{align*}{#1}\end{align*}}

\def\begequarr{\begin{eqnarray}}
\def\endequarr{\end{eqnarray}}
\def\begequarrs{\begin{eqnarray*}}
\def\endequarrs{\end{eqnarray*}}
\def\begarr{\begin{array}}
\def\endarr{\end{array}}
\def\begequ{\begin{equation}}
\def\endequ{\end{equation}}
\def\lab{\label}
\def\begdes{\begin{description}}
\def\enddes{\end{description}}
\def\begenu{\begin{enumerate}}
\def\begite{\begin{itemize}}
\def\endite{\end{itemize}}
\def\endenu{\end{enumerate}}

\def\lef[{\left[\begin{array}}
\def\rig]{\end{array}\right]}

\def\begcen{\begin{center}}
\def\endcen{\end{center}}
\def\begrem{\begin{remark}\rm}
\def\endrem{\end{remark}}
\def\begassum{\begin{assumption}}
\def\endassum{\end{assumption}}
\def\begassums{\begin{assumption*}}
\def\endassums{\end{assumption*}}
\def\begassu{\begin{ass}}
\def\endassu{\end{ass}}
\def\beglem{\begin{lemma}}
\def\endlem{\end{lemma}}
\def\begcor{\begin{corollary}}
\def\endcor{\end{corollary}}
\def\begfac{\begin{fact}}
\def\endfac{\end{fact}}

\def\ARC{{\it Annual Reviews in Control}}
\def\AJC{{\it Asian J. of Control}}

\def\TAC{{\it IEEE Trans. on Automatic Control}}

\def\AUT{{\it Automatica}}

\def\CSM{{\it IEEE Control Systems Magazine}}

\usepackage{color}

\def\bfori{{\bf ori}}
\def\bfnew{{\bf new}}

\journal{Systems \& Control Letters}

\begin{document}

\begin{frontmatter}

\title{Persistent Excitation is Unnecessary for On-line Exponential Parameter Estimation: A New Algorithm that Overcomes this Obstacle}


\author[IETR,ITMO]{M. Korotina\corref{corresponding}}
\cortext[corresponding]{Corresponding author}
\ead{marina.korotina@centralesupelec.fr}

\author[ITAM]{J. G. Romero}
\author[IETR,ITMO]{S. Aranovskiy}
\author[ITMO]{A. Bobtsov}
\author[ITAM]{R. Ortega}

\address[IETR]{IETR – CentaleSup\'elec, Avenue de la Boulaie, 35576 Cesson-S\'evign\'e, France}
\address[ITMO]{Faculty of Control Systems and Robotics, ITMO University, 197101 Saint-Petersburg, Russia}
\address[ITAM]{Departamento Acad\'emico de Sistemas Digitales, ITAM, R\'io hondo 1, Progreso Tizap\'an, 01080, Mexico City, Mexico}

\begin{abstract}
    In this paper we prove that it is possible to estimate on-line the parameters of a classical vector linear regression equation  $ {\bf Y}={\Omega} \theta$, where $ {\bf Y} \in \rea^n,\;{\Omega} \in \rea^{n \times q}$ are bounded, {\em measurable} signals and ${\theta} \in \rea^q$ is a constant vector of {\em unknown} parameters, even when the regressor $\Omega$ is not persistently exciting. Moreover, the convergence of the new parameter estimator is global and exponential and is given for both, continuous-time and discrete-time implementations. As an illustration example we consider the problem of parameter estimation of a linear time-invariant system, when the input signal is not sufficiently exciting, which is known to be a necessary and sufficient condition for the solution of the problem with standard gradient or least-squares adaptation algorithms. 
\end{abstract}

\begin{keyword}
Parameter estimation, Persistent excitation, Interval excitation, Dynamic regressor extension and mixing, Nonlinear filter
\end{keyword}

\end{frontmatter}

%
\section{Introduction and Problem Formulation}
\lab{sec1}
%
One of the central problems in control and systems theory, that has attracted the attention of many researchers for several  years, is the estimation of the parameters that appear in the mathematical model that describes the systems behavior, usually a differential or a difference equation.  A typical paradigm, which appears in system identification \cite{LJUbook}, adaptive control \cite{SASBODbook}, filtering and prediction  \cite{GOOSINbook}, reinforcement learning  \cite{LEWetal}, and in many other application areas, is when the unknown parameters and the measured data are linearly related in a so-called {\em linear regression equation} (LRE). Classical solutions for this problem are gradient and least-squares estimators. The main drawback of these schemes is that convergence of the parameter estimates relies on the availability of signal excitation, a feature that is codified in the restrictive assumption of persistency of excitation (PE) of the regressor vector. Moreover, their transient performance is highly unpredictable and only a weak monotonicity property of the estimation errors can be guaranteed.  

In recent years, various efforts to ease the PE requirement have been suggested, such as concurrent \cite{CHOetal}, or composite learning \cite{PANYU} that, in the spirit of off-line estimators, incorporate the monitoring of past data to build a stack of suitable regressor vectors. Another approach that has been extensively studied by the authors is the  dynamic regressor extension and mixing (DREM) parameter estimation procedure,  which was first proposed in \cite{ARAetaltac} for continuous-time (CT) and in \cite{BELetalsysid} for discrete-time (DT) systems. The construction of DREM estimators proceeds in two steps, first, the inclusion of a {\em free, stable, linear operator} that creates an extended matrix LRE. Second, a {\em nonlinear} manipulation of the data that allows to generate, out of an $q$-dimensional LRE,  $q$ {\em scalar}, and independent, LREs. DREM estimators have been successfully applied in a variety of identification and adaptive control problems, both, theoretical and practical ones, see \cite{ORTNIKGER,ORTetaltac} for an account of some of these results. 

A very important feature of the new concurrent and composite learning estimators is that parameter convergence is guaranteed under the extremely weak assumption of {\em interval excitation} (IE)  \cite{KRERIE}. This key property was also recently established for a version of DREM reported in \cite{GERetalsysid}, that has the additional feature of ensuring  convergence in {\em finite-time}---see also \cite[Propositions 6 and 7]{ORTetaltac}. A potential drawback of this DREM algorithm is that it relies on fixing the initial conditions of some filters, which may adversely affect the robustness of the estimator,  \cite[Remark 7]{ORTetaltac} and  \cite{ORTajc}.

In the recent paper  \cite{BOBetal} a procedure to generate, from a scalar LRE, new  scalar LREs where the {\em new regressor} satisfies some excitation conditions, even in the case when the original regressor is not exciting, was proposed. Instrumental for the development of the new adaptation algorithm is to borrow the key idea of the {\em  parameter estimation based observer} proposed in  \cite{ORTetalscl}, later generalized in \cite{ORTetalaut}, to generate the new LRE that includes some {\em free} signals. Then, applying the energy pumping-and-damping injection principle of \cite{YIetal}, we select these signals to guarantee some excitation properties of the new regressor. Unfortunately, to prove that the aforementioned excitation properties guarantee parameter convergence it is necessary to assume some {\em a priori} non-verifiable  conditions   \cite[Proposition 3]{BOBetal}---in particular the absolute integrability of a signal and a non-standard requirement on the limiting behavior of some of the components of the trajectories of the estimator.

In this paper we extend the DREM procedure and, in particular the results of \cite{BOBetal}, in several directions with our main contributions summarized as follows.
\begenu
\item[{\bf C1}] We give a definite answer to the question of ensuring that the  new regressor is PE assuming only the extremely weak condition of IE of the original vector regressor.  Towards this end, still abiding to the energy pumping-and-damping injection principle of \cite{YIetal}, we propose a new selection of the free signals of the LRE generator of   \cite{BOBetal} for which the exponential convergence proof can be completed without any additional assumptions.
\item[{\bf C2}] We illustrate our result with the important example of parameter identification of linear time-invariant (LTI) systems. It is well-known that a {\em necessary and sufficient} condition for global exponential convergence of the standard gradient (or least squares) estimators is the sufficient richness condition of the plants input signal \cite[Theorems 2.7.2 and 2.7.3]{SASBODbook}, which is equivalent to the  PE of the original regressor. We prove here that this condition is {\em not necessary}, and show that it is possible to exponentially estimate the parameters of the plant under the very weak assumption of IE of the original regressor.
\item[{\bf C3}] Motivated by the {\em practical} relevance of DT implementations we extend the LRE generator procedure of   \cite{BOBetal}, which was given for the CT case, to the DT case. Also, we propose the new DT signals that yield essentially the same results of CT mentioned in {\bf C1} and {\bf C2} above. 
\endenu

The remainder of the paper is organized as follows. Some background material of the Kreisselmeier regressor extension (KRE),  DREM estimators and the LRE generator procedure of   \cite{BOBetal} is given in Section \ref{sec2}.  In Section \ref{sec3} we present our main result discussed in {\bf C1} above. In Section \ref{sec4} we briefly discuss the results. 
Section \ref{sec5} presents the application to the parameter identification of LTI systems mentioned in {\bf C2}. Simulation results of the DT version of the result are presented in Section \ref{sec6}. The paper is wrapped-up with concluding remarks and future research in  Section \ref{sec7}.
To simplify the reading, a list of acronyms is given in the Appendix at the end of the paper.\\

\noindent {\bf Notation.} $I_n$ is the $n \times n$ identity matrix. $\intnum_{>0}$ and $\intnum_{\geq 0}$ denote the positive and non-negative integer numbers, respectively. For $x \in \rea^n$, we denote the Euclidean norm $|x|^2:=x^\top x$. CT signals $s:\rea_{\geq 0} \to \rea$ are denoted $s(t)$, while for DT sequences $s:\intnum_{\geq 0} \to \rea$ we use $s(k):=s(kT_s)$, with $T_s \in \rea_{> 0}$ the sampling time. The action of an operator $\mathcal H:\callinf \to \callinf$ on a CT signal $u(t)$ is denoted as $\mathcal H[u](t)$, while for an operator $\calh:\litcallinf \to \litcallinf$ and a sequence $u(k)$ we use  $\mathcal H[u](k)$. In particular, we define the derivative operator $p^n[u](t)=:{d^n u(t)\over dt^n}$ and the delay operator  $q^{\pm n}[u](k)=:u(k \pm n)$, where $n \in \intnum_{ >0}$.  When a formula is applicable to CT signals and DT sequences the time argument is omitted. 
%
\section{Background Material}
\lab{sec2}
%
In this section we present the following preliminary results which are instrumental for the development of our new results.
\begite
\item Derivation and properties of the KRE with the DREM estimator in CT \cite[Proposition 3]{ORTNIKGER} \cite[Proposition 1]{ARAetaltac} and in DT \cite[Proposition 3]{ORTetalaut21}.
\item   Generation of new LREs for CT \cite[Proposition 1]{BOBetal}\footnote{As explained in Section \ref{sec4} there is a slight modification of the $z(t)$ dynamics with respect to the one given in \cite[Proposition 1]{BOBetal}, namely the addition of a signal $u_4(t)$, that is introduced to simplify the proof of boundedness of $z(t)$.} and DT. Since the derivation of the DT LREs is reported here for the first time, we present also the proof of the proposition.
\item  Properties of the standard gradient estimator for the new LRE in CT \cite[Proposition 1]{ARAetaltac} and in DT \cite[Proposition 3]{ORTetalaut21}.
\endite

The following definitions will be used in the sequel.

\begin{definition}\em
\lab{def1}
A bounded signal $u \in \rea^{r \times s}$ is PE \cite{SASBODbook} if 
\begalis{
	&\int_t^{t+T_a} u(s) u^\top(s)  ds \ge C_a,	\\
}
for some $C_a>0$ and $T_a>0$  and for all $t \geq 0$ in CT and 
\begalis{
	&\sum_{j=k}^{k+k_b} u(j) u^\top(j) \ge C_b,
}
for some $C_b>0$ and $k_b  \in \intnum_{>0}$  and for all  $k  \in \intnum_{\geq 0}$ in DT.

It is said to be IE \cite{KRERIE,TAO} if 
\begalis{
	&\int_0^{t_c} u(s) u^\top(s)  ds \ge C_c	\\
}
for some  $C_c>0$ and $t_c>0$ in CT and
\begalis{
	&\sum_{j=0}^{k_d} u(j) u^\top(j) \ge C_d,
}
for some  $C_d>0$ and $k_d \in \intnum_{> 0}$ in DT.
\end{definition}

\begin{proposition} [Construction of the KRE]
\lab{pro1}\em
Consider the LRE
\begin{equation}\label{orilre}
    {\bf Y}={\Omega} \theta
\end{equation}
where $ {\bf Y} \in \rea^n,\;{\Omega} \in \rea^{n \times q}$ are bounded, {\em measurable} signals and ${\theta} \in \rea^q$ is a constant vector of {\em unknown} parameters. Fix the constants $\lambda >0$, $g  >0$,  $0 < \alpha <1$, and define the signals
\begali{
\nonumber
Z &=\calh[ \Omega^\top   {\bf Y}]\\
\nonumber
{ \Psi} &=\calh[  \Omega^\top   \Omega]\\
\nonumber
 \caly &= \adj\{ \Psi\} Z\\
\lab{eq1}
 \Delta &=\det\{ \Psi\},
}
where 
$$
\calh[u]=\left\{ \begarr{ccl} {g \over p + \lambda}[u](t) & \mbox{in} & CT \\ &&
\\ {g \over q - \alpha}[u](k) & \mbox{in} & DT, \endarr \right.
$$
and $\adj\{ \cdot\}$ denotes the adjugate matrix.
\begenu
\item[{\bf P1}] The signal $\Delta$ verifies
\begin{equation}
\label{delnonneg}
\Delta \geq 0. 
\end{equation}	
\item[{\bf P2}] The following implications are true \cite[Proposition 1]{ARAetaltac2}
\begin{equation}
\label{ieimpie}
\Omega \;  \mbox{in}  \left\{ \begarr{ccl} \mbox{IE} \\ &&\\ \mbox{PE}  \endarr \right. \Rightarrow \quad	\Delta  \;  \mbox{in}  \left\{ \begarr{ccl} \mbox{IE} \\ &&\\ \mbox{PE}  \endarr \right.
\end{equation}
\item[{\bf P3}] The $q$ scalar LREs
\begequ
\label{scalre}
 \caly_i= \Delta   \theta_i,\;i \in \{1,2,\dots,q\},
\endequ
hold.
\endenu
\end{proposition}
\begin{proposition}[Generation of new LREs] \em
\lab{pro2}
Consider the scalar LREs  \eqref{scalre}.\footnote{To simplify the notation we omit the subindex $i$ in the proposition.} Define the dynamic extension
\begsubequ
\lab{dynext}
\begali{
\lab{z}
\mathfrak{d}[z] & =  u_{2} \caly+  u_{3}  z+  u_{4},\; z(0)=0 \\	
\lab{xi}
	 \mathfrak{d}[\xi] & =  A  \xi +  b,\;  \xi(0)=\col(0,0) 	 \\
\lab{phi}
\mathfrak{d}[  \Phi]& =  A  \Phi,\;  \Phi(0)=\col(1,0),
}
\endsubequ
where the operator $\mathfrak{d}[\cdot]$ is defined as
\begequ
\lab{oped}
\mathfrak{d}[u]=\left\{ \begarr{ccl} p[u](t) & \mbox{in} & CT \\ &&\\ q[u](k) & \mbox{in} & DT \endarr \right.
\endequ
and we defined
\begequ
\lab{ab}
 A :=\begmat{A_{11} &  u_{1} \\  u_{2}  \Delta &  u_{3}},\;  b:=\begmat{- u_{1}  z\\  u_{4}},
\endequ
with $ u_{i} \in \rea$, $i=1,\dots,4$, {\em arbitrary} signals and 
\begequ
\lab{a11}
A_{11}=\left\{ \begarr{ccl} 0 & \mbox{in} & CT \\ &&\\ 1 & \mbox{in} & DT \endarr \right.
\endequ
The new LRE
\begin{equation}
\label{newlre}
Y  = {   \Phi}_{2}  \theta, 
\end{equation}
holds with
$$
{Y}:= z-  \xi_{2}.
$$
\end{proposition}
\begin{proof}
 
 [DT version]\footnote{The proof of the CT case may be found in \cite{BOBetal}.} Notice that, since $\theta$ is constant, we can write
\begequ
\lab{thek}
	\theta(k+1)  =\theta(k)+ u_1(k) [z(k)-z(k)], \theta(0)=\theta.
\endequ
Combining \eqref{z} and \eqref{thek}, and using \eqref{scalre}, we can write the ``virtual" LTV system
\begequ
\label{xk}
x(k+1) = A(k)x(k) + b(k),
\endequ
with $x(k) :=\col(\theta(k), z(k))$ and initial conditions
\begequ
\lab{icx}
x(0) =\begmat{\theta \\ 0}.
\endequ

Define the error signal
\begequ
\label{e}
e(k):=\xi (k)- x(k),
\endequ
which satisfies $e(k+1)=A(k) e(k)$. Consequently, from \eqref{e} and the properties of the signals  $\Phi(k)$ defined in \eqref{phi}, we get
\begali{
\nonumber
x(k) &=\xi(k) - \Phi(k) e_1(0)\\
\lab{xxik}
&=\xi (k)+ \Phi(k) \theta
}
where, to get the second identity, we took into account \eqref{icx} and the initial conditions in \eqref{xi}.

Now
$$
\begmat{\caly(k)\\ z(k)}=\begmat{ \Delta(k) &0 \\ 0 &1}x(k)=\begmat{ \Delta(k) &0 \\ 0 &1}\Big(\xi (k)+\Phi(k)\theta \Big).
$$
The proof is completed defining
\begequ
\lab{calyk}
Y(k) =\begmat{ Y_1(k)\\  Y_2(k)}:=\begmat{\caly(k)\\ z(k)}-\begmat{ \Delta(k)\xi_1(k) \\ \xi_2(k)}.
\endequ
  \end{proof}
  
 \begin{proposition}[Convergence properties of the gradient estimator]
\lab{pro3}\em
Consider the scalar LRE  \eqref{newlre} of  Proposition \ref{pro2} with the gradient estimator
\begequ
\label{eq4}
\mathfrak{p}[\hat{ \theta}] =  \gamma   \Phi_{2} \left( Y -  \Phi_{2} \, \hat { \theta}\right),
\endequ
with $ \gamma>0$ and the operator $\mathfrak{p}[\cdot]$ is defined as 
\begequ
\lab{opep}
\mathfrak{p}[u]=\left\{ \begarr{ccl} p[u](t) & \mbox{in} & CT \\ &&\\ (q-1)[u](k) & \mbox{in} & DT \endarr \right.
\endequ
The following equivalences are true:
\begalis{
 \hat{ \theta} \to \theta \; &  \Leftrightarrow \; \Phi_{2} \not \in \left\{ \begarr{ccl} \call_2 &  \mbox{in} & CT \\ &&\\ \ell_2& \mbox{in} & DT \endarr \right.\\
\hat{ \theta}\to \theta,\;(exp)\;  & \Leftrightarrow \;	\Phi_{2}  \in \mbox{PE}.
}
\end{proposition}
%
\section{Main Result}
\lab{sec3}
%
In this section we give the main result of the paper, namely the selection of the signals $u_i,\;i=1,\dots,4$, in  \eqref{dynext}-\eqref{ab} that ensure the new regressor $\Phi_2$ is PE under the very weak assumption of IE of the regressor $\Omega$ of the original LRE \eqref{orilre}. In the DT case, it is necessary to impose an additional assumption on a tuning parameter.

\begin{proposition}\em
\lab{pro4}
Consider the LRE \eqref{orilre} and the KRE construction of Proposition \ref{pro1}. Assume $\Omega$ is IE. Consider the dynamics \eqref{dynext}-\eqref{ab} with the signals
\begequ
\lab{u}
u(t)=\begmat{-\mu\Delta(t) \Phi_1(t)\\ \mu \Phi_1(t)\\   -\tilde V(t) \\  [\tilde V(t)-\mu]  z(t)},
\endequ
in CT and 
\begequ
\lab{uk}
u(k)=\begmat{-T\mu\Delta(k) \Phi_1(k)\\ T\mu \Phi_1(k)\\  1 -T\tilde V(k)\\  [T\tilde V(k)-b]z(k)}
\endequ
in DT, where
\begequ
\lab{tilv}
   \tilde V := \frac{1}{2}\left(\Phi_1^2 + \Phi_2^2\right)-\beta,
\endequ
with $0 < b < 1$, $\beta>\frac{1}{2}$,  $\mu >0$ and $T  >0$ is a small number such that we can assume
\begequ
\lab{tsqu}
   T^2 \approx 0.
\endequ 

\begenu
	\item[{\bf F1}]  The signals $z$, $\xi$ and $\Phi$ are bounded.
	\item[{\bf F2}]  $\Phi_2$ is PE.
\endenu
\end{proposition}

\begin{proof}
The proof proceeds in the following three steps.
\begite
\item Proof of boundedness of $\Phi$.
\item Proof of PE of $\Phi_2$.
\item Proof of boundedness of $z$ and $\xi$.
\endite
Although the arguments for the CT and the DT case are similar, for the sake of clarity, we present them in separate subsections whenever needed.
\subsubsection*{(i) Proof of boundedness of $\Phi(t)$}
Replacing \eqref{u} in \eqref{dynext}-\eqref{ab} yields
\begsubequ
\lab{dotphi}
\begin{align}
	\dot{\Phi}_1(t) &= -\mu\Delta(t)\Phi_2(t)\,\Phi_1(t), \label{eq:phi1:beta} \\
	\dot{\Phi}_2(t) &= \mu\Delta(t)\Phi_1^2(t) -\tilde V(t)  \Phi_2(t). \label{eq:phi2:beta}
\end{align}
\endsubequ
From \eqref{tilv} and the equations of $\Phi(t) $ above we immediately get
\begequ
\lab{dottilv}
\dot {\tilde V}(t)=-\Phi^2_2(t) \tilde V(t),
\endequ
from which we conclude the invariance of the set
\begequ
\lab{ome}
\Omega:=\{\Phi \in \rea^2|  \frac{1}{2}\left(\Phi_1^2 + \Phi_2^2\right)=\beta\}.
\endequ

Now, invoking the initial condition constraint $\Phi_1(0) = 1$, and $\Phi_2(0) =0$, we have
$$
\begin{aligned}
& \hal(\Phi^2_1(0)+\Phi^2_2 (0))  =\hal \\
& \Leftrightarrow \tilde V(0)  + \beta = \hal \\
& \Leftrightarrow \tilde V(0)  = \hal  - \beta \\
&  \Rightarrow  \tilde V(0)  < 0,
\end{aligned}
$$
where we used $\beta>\frac{1}{2}$ to get the last implication. The latter inequality implies that the trajectory starts {\em inside} the disk delimited by the set $\Omega$.  This, together with the invariance of the set implies, that the {\em whole trajectory} $(\Phi_1(t), \Phi_2 (t))$ is inside this disk, that is,
\begequ
\lab{insdis}
 \tilde V(t) \leq 0,\;\forall t {\ge} 0. 
\endequ
Replacing the bound above in \eqref{dottilv} we have that $\dot {\tilde V}(t) \geq 0$ from which we conclude that $\tilde V(t)$ is {\em non-decreasing}, hence 
\begequ
\lab{lowboutilv}
\tilde V(t) \geq  \tilde V(0)  = \hal  - \beta.
\endequ
Combining the bounds \eqref{insdis} and \eqref{lowboutilv} we conclude that
$$
1 \leq \Phi^2_1(t)+\Phi^2_2 (t) \leq 2 \beta.
$$
The bounds above can be further sharpened as follows. From the constant $\mu(t) \ge 0$, \eqref{eq:phi2:beta}, \eqref{insdis},  \eqref{delnonneg}, and recalling the initial condition $\Phi_2(0)=0$, it follows that
$\dot{\Phi}_2(t) \ge 0$ and, consequently, $\Phi_2(t) \ge 0$ for all $t$. Moreover, 
\[
	\dot{\Phi}_1(t) = -\mu\Delta(t)\Phi_2(t)\,\Phi_1 \le 0,
\]
hence $\Phi_{1}(t)$ is not increasing and, recalling that $\Phi_1(0)=1$, it follows that  $0\le \Phi_1(t) \le 1$. 

In summary, the whole trajectory $\Phi(t)$ lives in the gray section indicated in Fig. \ref{fig1}.

\begin{figure}[h]
    \centering
	\includegraphics[width = 0.4\textwidth]{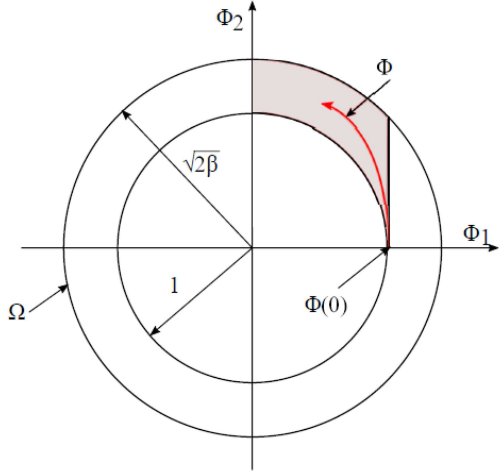}
    \caption{Behavior of the trajectory $\Phi$.}
    \label{fig1}
\end{figure}

\subsubsection*{(ii) Proof of boundedness of $\Phi(k)$} 
Replacing \eqref{uk} in  \eqref{dynext}-\eqref{ab} yields the dynamics
\begequ
\lab{clolook}
\Phi(k+1)= \begmat{1 & - T \mu \Delta(k) \Phi_1(k)  \\ T \mu \Delta(k) \Phi_1(k) & 1-T \tilde V(k) } \Phi(k).  
\endequ
Hence, computing
\begalis{
 |\Phi(k+1)|^2 =&  \Phi^\top (k) \begmat{    1  &  T \mu\Delta (k)\Phi_1(k)  \\
-T \mu\Delta(k) \Phi_1 (k)&  1-  T \tilde V(k) } \begmat{    1  & - T \mu\Delta(k) \Phi_1(k)  \\
T \mu\Delta (k)\Phi_1(k) &  1-  T \tilde V(k) }   \Phi(k)  \\
=& \Phi^\top (k) \begmat{    1  &  0  \\ 0 &  1- 2 T \tilde V(k)}   \Phi(k) + T^2 \Phi^\top (k) \begmat{   \mu^2 \Delta^2(k) \Phi^2_1(k)    &  -\mu\Delta(k) \Phi_1(k)   \\- \mu\Delta(k) \Phi_1 (k)&   \mu^2 \Delta^2(k) \Phi^2_1(k) +\tilde V^2(k) }   \Phi(k)\\
=& |\Phi (k)|^2 - 2 T \tilde V(k)  \Phi^2_2(k) + T^2 \Phi^\top (k) \begmat{   \mu^2 \Delta^2(k) \Phi^2_1(k)    &  -\mu\Delta(k) \Phi_1(k)   \\- \mu\Delta(k) \Phi_1 (k)&   \mu^2 \Delta^2(k) \Phi^2_1(k) +\tilde V^2(k) }   \Phi(k).
}
Invoking that $T^2 \approx 0$ we obtain 
\begequ
\lab{tilvinck}
\tilde V(k+1)= \tilde V(k) -T \tilde V (k) \Phi^2_2 (k)
\endequ
from which we conclude the invariance of the set \eqref{ome}.

Following {\em verbatim} the reasoning carried out in CT we conclude that the {\em whole trajectory} $(\Phi_1(k), \Phi_2 (k))$ is inside  the disk described by the set $\Omega$, that is,
\begequ
\lab{insdisk}
 \tilde V(k) \leq 0,\;\forall k \in \intnum_{\geq 0}. 
\endequ
Replacing the bound above in \eqref{tilvinck} we have that  $\tilde V(k)$ is {\em non-decreasing}, hence 
\begequ
\lab{lowboutilvk}
\tilde V(k) \geq  \tilde V(0)  = \hal  - \beta.
\endequ
Combining the bounds \eqref{insdisk} and \eqref{lowboutilvk} we conclude that
$$
1 \leq \Phi^2_1(k)+\Phi^2_2 (k) \leq 2 \beta.
$$

As done in CT the bounds above can be further sharpened as follows. From \eqref{clolook} we have that
$$
\Phi_2(k+1)=  [1 -T \tilde V(k) ] \Phi_2(k)+T \mu\Delta(k) \Phi^2_1(k).
$$
From  $\mu > 0$, $\Delta(k) \ge 0$ and \eqref{insdis} it follows that $\Phi_2(k)$ is {\em non-decreasing}. Moreover, recalling the initial condition $\Phi_2(0)=0$, it follows that $\Phi_2(k) \ge 0$ for all $k \in \intnum_{\geq 0}$. 

Now, from \eqref{clolook} we also have that 
$$
\Phi_1(k+1)= [1  - T \mu\Delta(k) \Phi_2(k)] \Phi_1(k).
$$
Since, the term in brackets is not bigger than one, we have that the sequence  $\Phi_{1}(k)$ is {\em non-increasing} and, recalling that $\Phi_1(0)=1$, it follows that  $0\le \Phi_1(k) \le 1$. 

In summary, the whole trajectory $\Phi(k)$ lives in the gray section indicated in Fig. \ref{fig1}.

\subsubsection*{(iii) Proof of PE of $\Phi_2$}
The assumption that $ \Delta(t)$ in IE implies that there exists a $t_0 \in (0,t_c]$ such that $ \Delta(t_0) >0$, which in turn implies that there exists a $t_\rho>0$ such that $\rho:=\Phi_2(t_\rho)>0$. Since we proved above that $ \Phi_2(t)$ is non-decreasing we have that
$$
 \Phi_2(t) \geq \rho>0,\;\forall t \geq t_\rho.
$$
Consequently,
\[
	\operatorname*{lim\,inf}_{t\to \infty}\Phi_2(t) >0,
\]
and $\Phi_2(t)$ is PE.\footnote{It is well-known that a scalar signal (with a bounded derivative) that {\em does not} converge to zero is PE.}

Exaclty the same arguments can be used in DT to prove that
\[
	\operatorname*{lim\,inf}_{k \to \infty}\Phi_2(k) >0,
\]
hence $\Phi_2(k)$ is PE. 

\subsubsection*{(iii) Boundedness of $z$ and $\xi$}
From \eqref{xxik}---and the equivalent relation in CT \cite{BOBetal}---we have that
$$
x=\begmat{\theta \\ z}=\xi + \Phi \theta.
$$
Since we proved that $ \Phi$ is bounded, to establish boundedness of $\xi$ it suffices to prove that $z$ is bounded. 
Towards this end, we replace \eqref{u} or  \eqref{uk} in the $z$ dynamics of  \eqref{z} to get
\begalis{
	\dot{z}(t) & = -\tilde V(t)  z (t)+ \mu\Phi_1 (t)\caly(t)	+u_4(t)\\
	&= -  \mu z(t)+ \mu\Phi_1(t) \caly(t),
}
in CT and
\begalis{
	{z}(k) & =[1 -T\tilde V(k)]  z(k) + T \mu\Phi_1(k) \caly(k)	+u_4(k)\\
	&= -  (1-b) z(k)+ T \mu\Phi_1(k) \caly(k),
}
in DT. In both cases, we are dealing with asymptotically stable LTI filters with bounded input, 

completing the proof.
\end{proof}
%
\section{Discussion}
\lab{sec4}
%
The following remarks are in order.\\

\bul
The main message of Proposition \ref{pro4} is that it is possible to estimate the parameters of a classical vector LRE \eqref{orilre} even when the regressor $\Omega$ is not PE---the convergence of the new parameter estimator being global and exponential.\footnote{Additional properties of the DREM estimator, like element-by-element monotonicity of the parameter errors, may be found in \cite{ORTetaltac}.}\\

\bul 
The choice of the signals $u(k)$ given in \eqref{uk} is motivated by the CT dynamics \eqref{dotphi}. Indeed, the DT dynamics of $\Phi(k)$ given in \eqref{clolook} is the {\em Euler approximation} of \eqref{dotphi}. It is well-known \cite{STOBULbook} that the Euler approximation is a numerical integration method of order one whose global approximation error
is $\calo(T^2)$.\footnote{$f(t,T)$ is  ``big o of $T^2$" if and only if  $|f(t,T)| \leq C T^2$ with $C$ a constant independent of $t$ and $T$.} This explains our need to impose the assumption \eqref{tsqu} in our stability analysis.\\

\bul Although it is possible to consider other (higher order) discretization methods of the $\Phi(t)$ dynamics  \eqref{dotphi}, the resulting discretized dynamics cannot be matched with the $A(k)$ matrix given in \eqref{ab} due to the fact that---as seen in  \eqref{a11}---it is {\em necessary} to have the term $A_{11}(k)=1$. A condition that stymies the selection of a more precise discretization method. \\ 

\bul
Another alternative to remove the undesirable  assumption \eqref{tsqu} is to directly pose a regulation problem for the DT system identified in Proposition \ref{pro2}, that is
$$
\Phi(k+1)= \begmat{1 & u_1(k)  \\  u_2(k) \Delta(k) & u_3(k) } \Phi(k).  
$$
The task is to select  the signals $u_i(k),\;i=1,2,3$, that insure boundedness of all signals and that $\Phi_2(k)$ is PE. Unfortunately,  this a highly complicated nonlinear control problem with non-standard regulation objectives.\\

\bul
In \cite{BOBetal} the proof of boundedness of the signal $z(t)$ is quite involved and requires the additional of an unverifiable absolute integrability assumption \cite[Equation (14)]{BOBetal}. This is due to the fact that the new free signal $u_4(t)$  in the vector $b(t)$ in \eqref{ab}, was not included in  \cite{BOBetal}. It is cleat that the addition of this signal does not affect the main result, and trivializes the proof of  boundedness of $z(t)$.
%
\section{Application to Identification of CT Systems in Unexcited Conditions}
\lab{sec5}
%
To illustrate the result of Proposition \ref{pro4}, we consider in this section the problem of parameter estimation of an CT LTI system and choose, as an example, the system:
\begin{equation}\label{eq:y_BAu}
    y_p(t) = \frac{B(p)}{A(p)}[u_p](t) = \frac{b_1p + b_0}{p^2 + a_1p + a_0}[u_p](t),
\end{equation}
where $u_p(t)\in \rea$ and $y_p(t) \in \rea$ are the control and output signals, respectively. Following the standard procedure \cite[Subsection 2.2]{SASBODbook} we derive the LRE \eqref{orilre} for the system \eqref{eq:y_BAu} as follows
\begin{equation}\label{eq:phi_parametrization}
     {\bf Y}(t):=y_p(t),\;{\Omega}(t) :=\begin{bmatrix}
    \frac{F(p) B(p)}{A(p)} \\
    F(p)
    \end{bmatrix} [u_p](t) , \  F(p):=\frac{1}{\lambda(p)}\begin{bmatrix}
    1 \\
    p \\
    \vdots \\
    p^{n-1}
    \end{bmatrix}, \ \theta :=\begin{bmatrix}
    \lambda_{0}-a_{0} \\
    \vdots \\
    \lambda_{n-1}-a_{n-1} \\
    b_{0} \\
    \vdots \\
    b_{n-1}
    \end{bmatrix},
\end{equation}
with $\lambda(p)=\sum_{i=0}^{n} \lambda_{i} p^{i}, \lambda_{n}=1$, an arbitrary Hurwitz polynomial.\\

We consider the following simulation scenarios. 
\begenu
\item[{\bf S1}] Estimation of the vector $\theta$ with the standard gradient estimator
\begin{equation}
\label{eq:Grad_est}
    \dot{\hat{\theta}}(t)=\Gamma {\Omega}(t)\left[ {\bf Y}(t)-{\Omega}^\top(t) \hat{\theta}(t)\right], \ \Gamma>0.
\end{equation}
\item[{\bf S2}]  Estimation of the parameters $\theta_i$ using the scalar regression form \eqref{scalre} obtained via the KRE and DREM of Proposition \ref{pro1}, that is
\begin{equation}
\label{eq:gradientDREM}
    \dot{\hat{\theta}}_{i}(t) =\gamma_{i} \Delta(t) \left[\caly_i (t)-\Delta(t)  \hat{\theta}_{i}(t) \right],\;\gamma_i > 0.
\end{equation}
\item[{\bf S3}]  Estimation of the parameters $\theta_i$ using the new scalar regression form  \eqref{newlre} obtained via the LRE generator of Proposition \ref{pro4}, that is
\begin{equation}\label{eq:gradAlgNew}
	\dot{\hat{\theta}}_i (t)= \gamma_i\Phi_{2}(t)\left(Y_i (t)-\Phi_{2}(t)\hat{\theta}_i(t)\right),\;\gamma_i > 0.
\end{equation}
\item[{\bf S4}] Simulation of the three estimators above for a sufficiently rich input signal  
\begequ
\lab{upa}
    u_{pa}(t) = \sin(2\pi t) + \cos(3t),
\endequ
and for an input signal that is not sufficiently rich, but generates a regressor $\Omega$ which is IE, namely
\begequ
\lab{upb}
    u_{pb}(t) = e^{-2t}+e^{-1.5t}.
\endequ
\endenu

The following remarks concerning the theoretical results are in order 
\begite
\item For the sufficiently rich signal \eqref{upa} the three estimators yield consistent estimates.
\item For the not sufficiently rich signal \eqref{upb} the first and second estimators will not generate consistent estimates. For the first estimator this follows from the fact that, as shown in \cite[Theorems 2.7.2 and 2.7.3]{SASBODbook} sufficient richness of the plants input signal  is equivalent to PE of the regressor $\Omega(t)$. Regarding the DREM estimator with the regressor $\Delta(t)$, it was shown in  \cite[Proposition 2]{aranovskiy2019parameter} that DREM alone cannot relax the PE condition in the system identification problem. Hence, sufficient richness is {\em necessary} for parameter convergence.
\item On the other hand, the result of Proposition \ref{pro4} ensures that DREM with the new LRE will ensure convergence even for the input signal  \eqref{upb}.
\endite

The simulations were carried out for the system studied in  \cite[Section 5]{aranovskiy2019parameter}, that is $(a_0,a_1,b_0,b_1)=(2,1,1,2)$ and we choose $\lambda_1 = 20$ and $\lambda_0 = 100$. This yields $ \theta = \begin{bmatrix} 98 & 19 & 1 & 2 \end{bmatrix}.$

For all estimators we set $\hat{\theta}_i(0)=0$. The gain matrix for the first algorithm \eqref{eq:Grad_est} is 
\[
    \Gamma = 100 \operatorname{diag}(100,50,30,10).
\]
For the estimators \eqref{eq:gradientDREM} and\eqref{eq:gradAlgNew} we choose $\gamma_i = 1$, $i=1,\dots,4$. The parameters of the KRE \eqref{eq1} are $g = 100$, $\lambda = 30$. For the new LRE we choose $\beta = 0.51$. 

The simulation results, which corroborate the claims above, are shown in Figures \ref{fig:u1_thetaGRAD}-\ref{fig:u2_thetaDREMIMP}. Notice, in particular, that for the input signal  \eqref{upb} only DREM with the new LRE ensures convergence.
 
 \begin{figure}[h]
     \centering
     \includegraphics{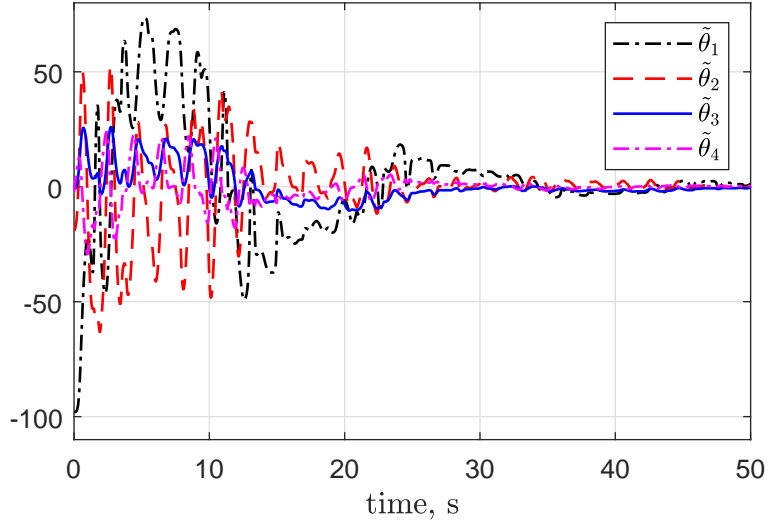}
     \caption{$\tilde{\theta}(t)$ with gradient algorithm \eqref{eq:Grad_est} and $u_{pa}(t)$}
     \label{fig:u1_thetaGRAD}
 \end{figure}
 
  \begin{figure}[h]
     \centering
     \includegraphics{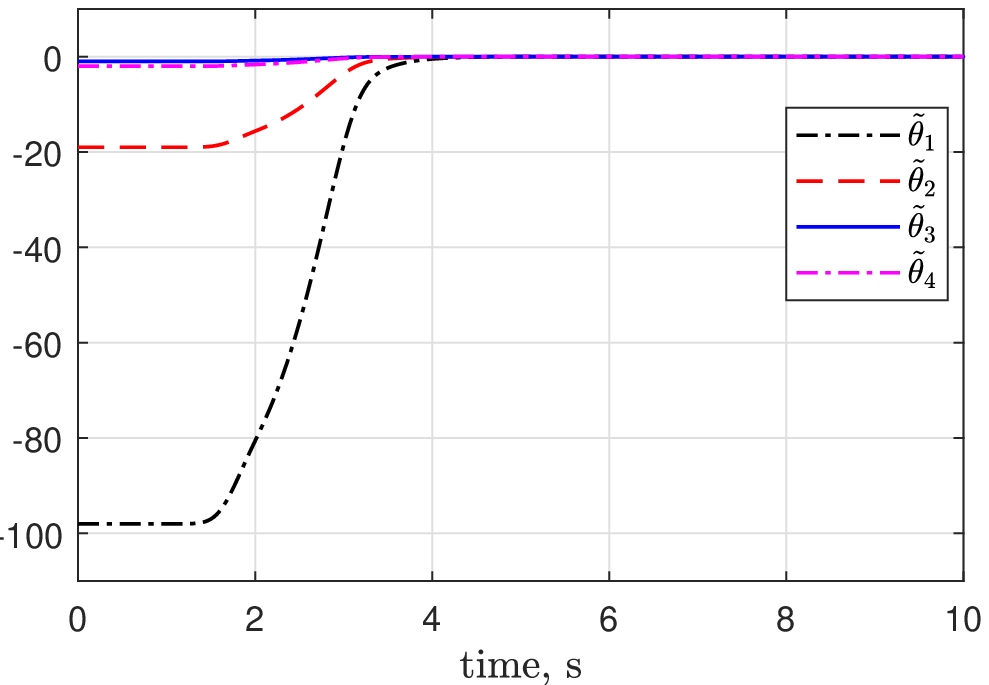}
     \caption{$\tilde{\theta}(t)$ with DREM procedure, estimator \eqref{eq:gradientDREM} and $u_{pa}(t)$}
     \label{fig:u1_thetaDREM}
 \end{figure}
 
 \begin{figure}[h]
     \centering
     \includegraphics{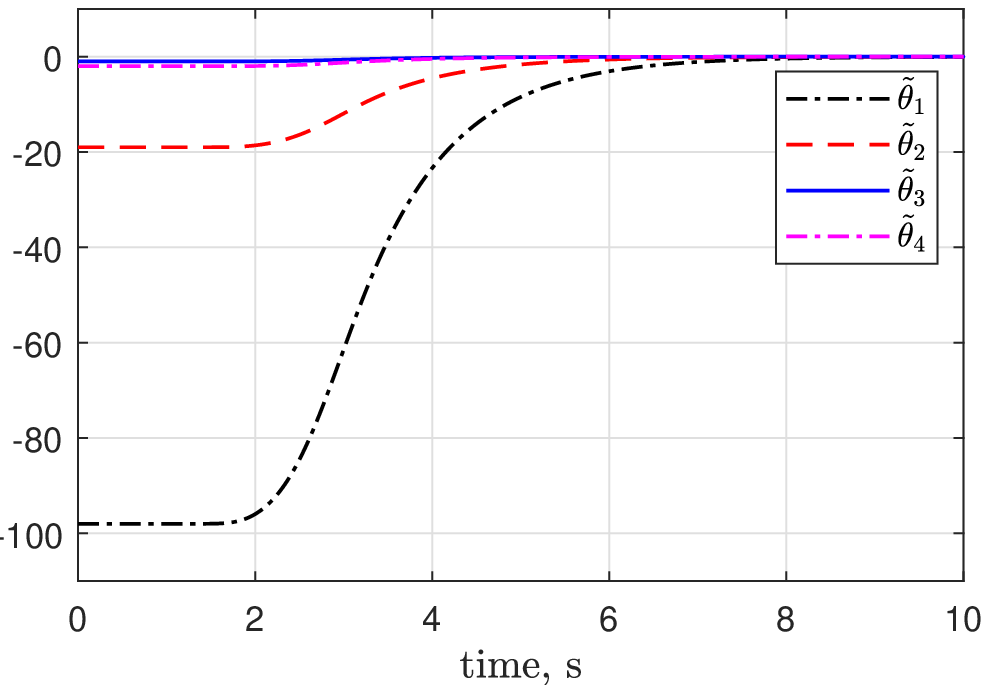}
     \caption{$\tilde{\theta}(t)$ with DREM procedure and new LRE, estimator \eqref{eq:gradAlgNew} and $u_{pa}(t)$}
     \label{fig:u1_thetaDREMIMP}
 \end{figure}
 
 \begin{figure}[h]
     \centering
     \includegraphics{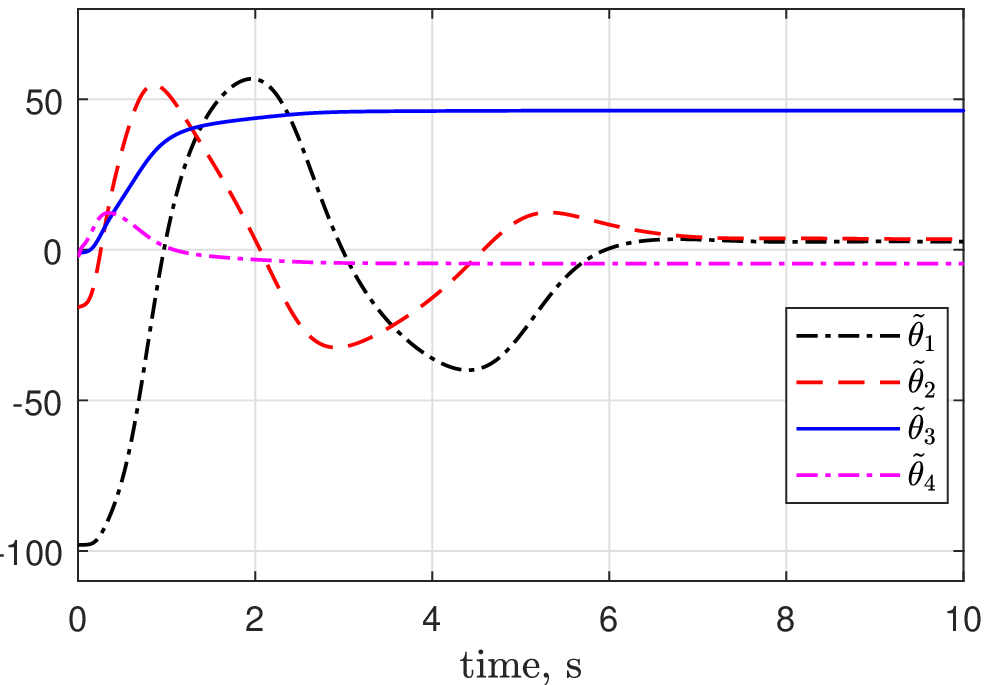}
     \caption{$\tilde{\theta}(t)$ with gradient algorithm \eqref{eq:Grad_est} and $u_{pb}(t)$}
     \label{fig:u2_thetaGRAD}
 \end{figure}
 
  \begin{figure}[h]
     \centering
     \includegraphics{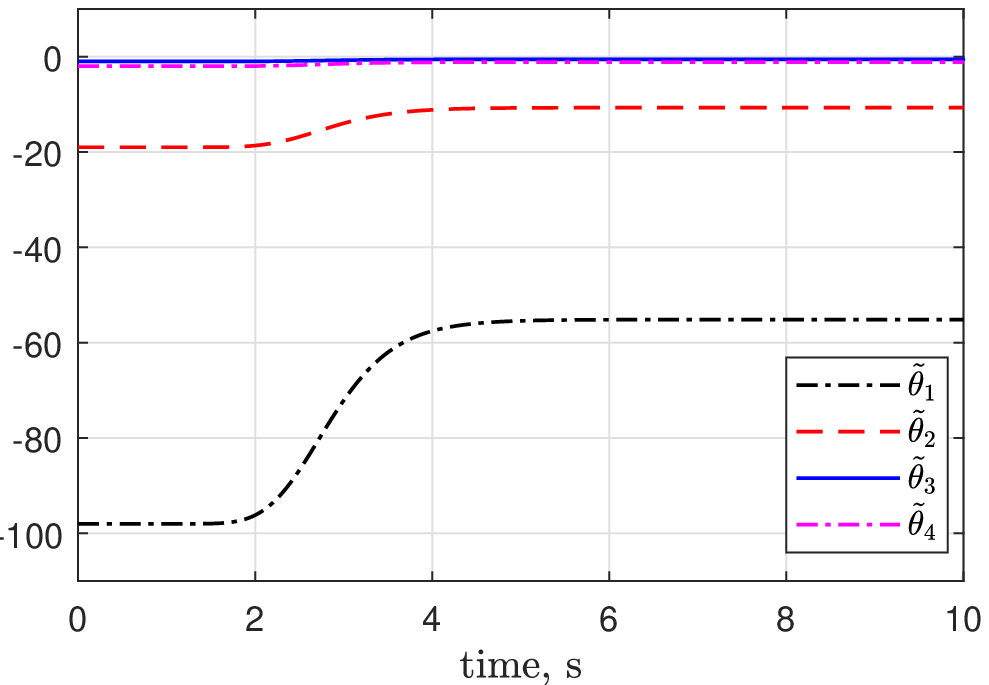}
     \caption{$\tilde{\theta}(t)$ with DREM procedure, estimator \eqref{eq:gradientDREM} and $u_{pb}(t)$}
     \label{fig:u2_thetaDREM}
 \end{figure}
 
 \begin{figure}[h]
     \centering
     \includegraphics{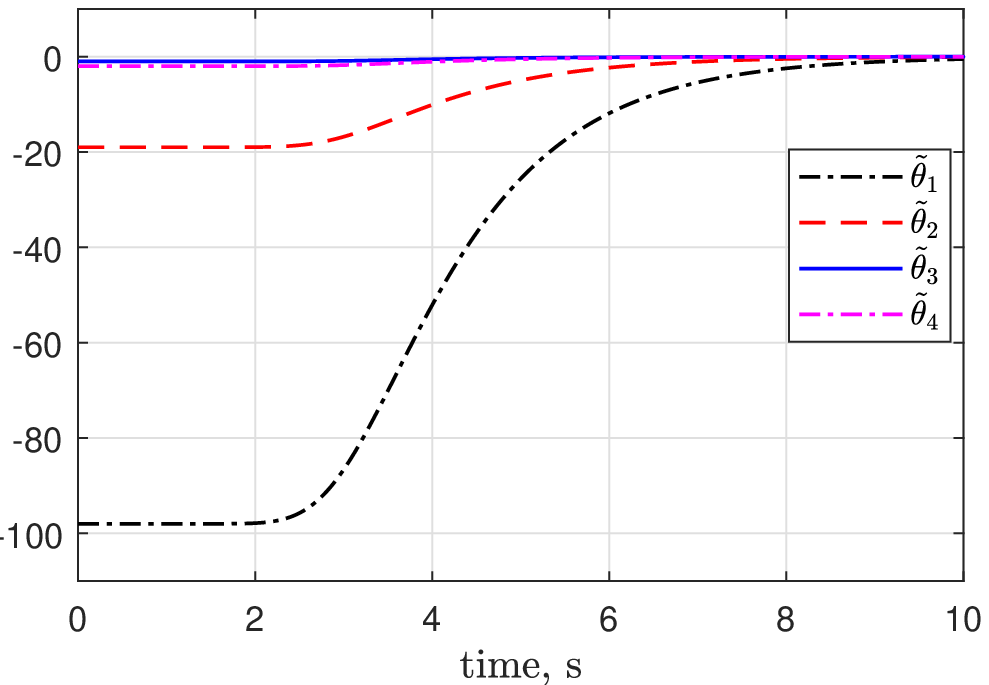}
     \caption{$\tilde{\theta}(t)$ with DREM procedure and new LRE, estimator \eqref{eq:gradAlgNew} and $u_{pb}(t)$}
     \label{fig:u2_thetaDREMIMP}
 \end{figure}
 
 To test the robustness of the various estimators bounded noise was added to the output signal, as shown in Fig. \ref{fig:y2_noise}), for the case of the input  $u_{pb}(t)$.
 
 \begin{figure}[h]
     \centering
     \includegraphics{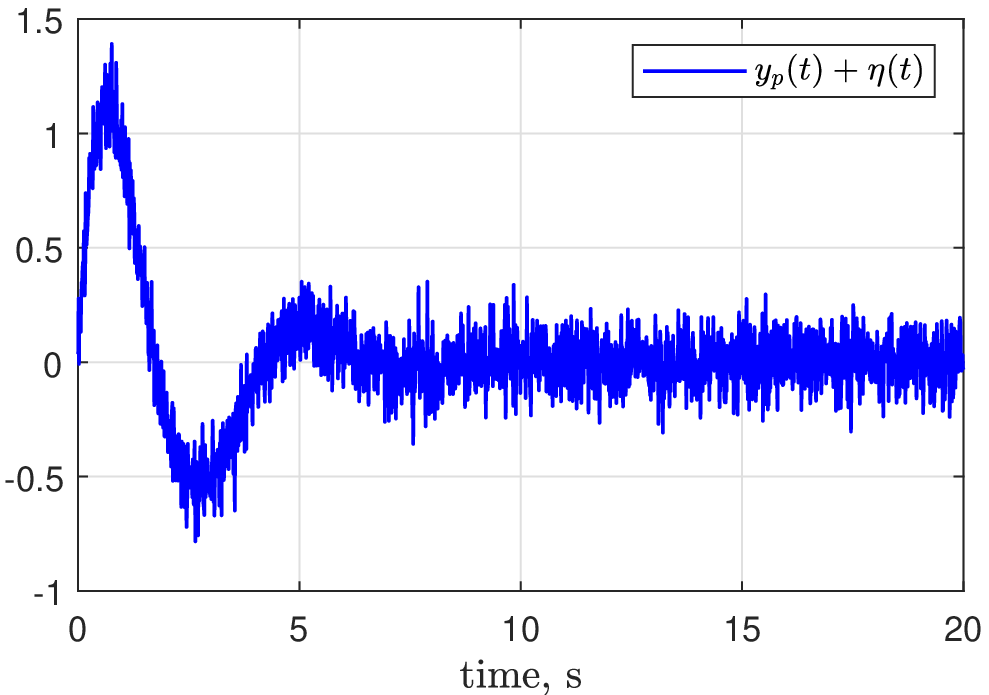}
     \caption{$y_p(t)$ with additive noise $\eta(t)$}
     \label{fig:y2_noise}
 \end{figure}

As expected, in this case none of the estimators ensures that the error $\tilde{\theta}(t)$  converges to zero, as shown in Figs. \ref{fig:u2_thetaGRAD_noise}-\ref{fig:u2_thetaDREMIMP_noise}. However, notice that the gradient algorithm \eqref{eq:Grad_est} actually diverges. On the other hand, while the steady state error of the estimator   with DREM \eqref{eq:gradientDREM} is quite large, the one of \eqref{eq:gradAlgNew} with the new LRE, is negligible---illustrating the robustness to additive noise of the the new scheme.

 \begin{figure}[h]
     \centering
     \includegraphics{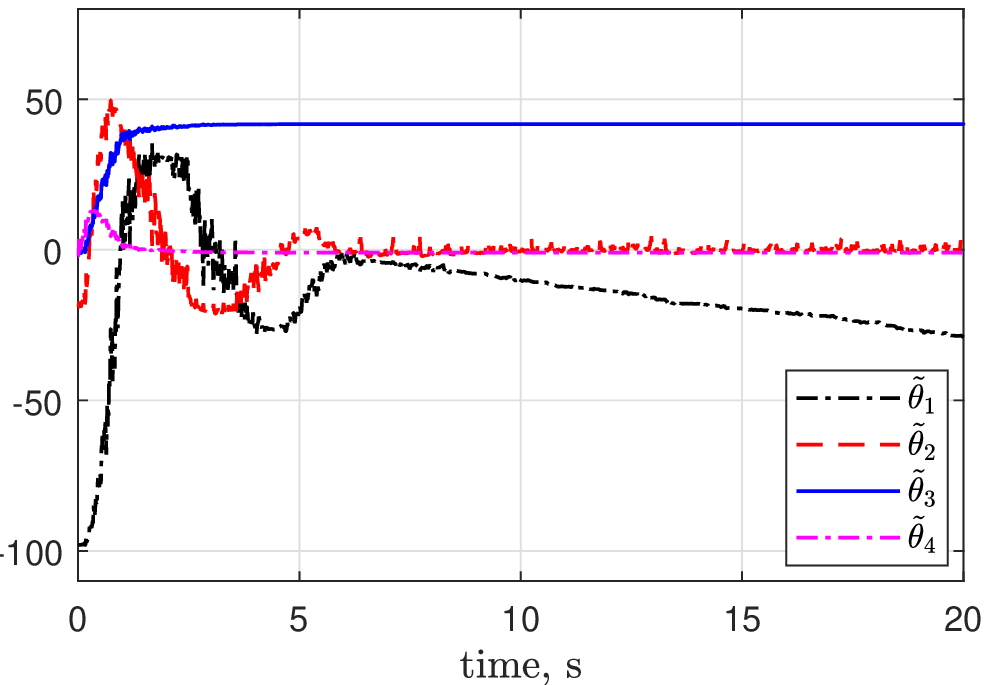}
     \caption{$\tilde{\theta}(t)$ (under the influence of $\eta(t)$) with gradient algorithm \eqref{eq:Grad_est} and $u_{pb}(t)$}
     \label{fig:u2_thetaGRAD_noise}
 \end{figure}
 
  \begin{figure}[h]
     \centering
     \includegraphics{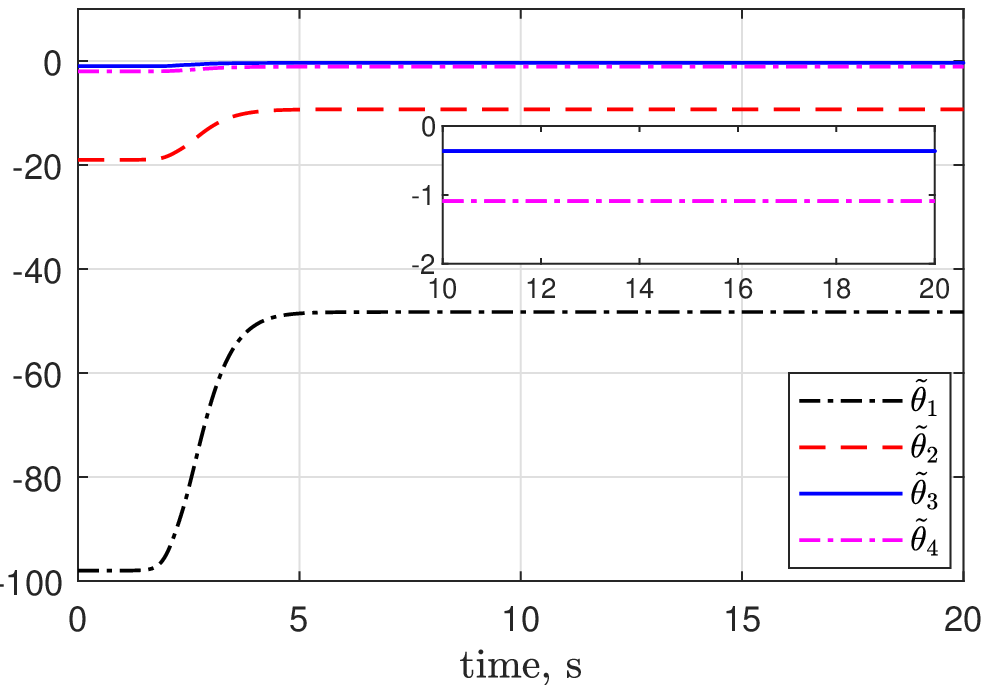}
     \caption{$\tilde{\theta}(t)$ (under the influence of $\eta(t)$) with DREM procedure, estimator \eqref{eq:gradientDREM} and $u_{pb}(t)$}
     \label{fig:u2_thetaDREM_noise}
 \end{figure}
 
 \begin{figure}[h]
     \centering
     \includegraphics{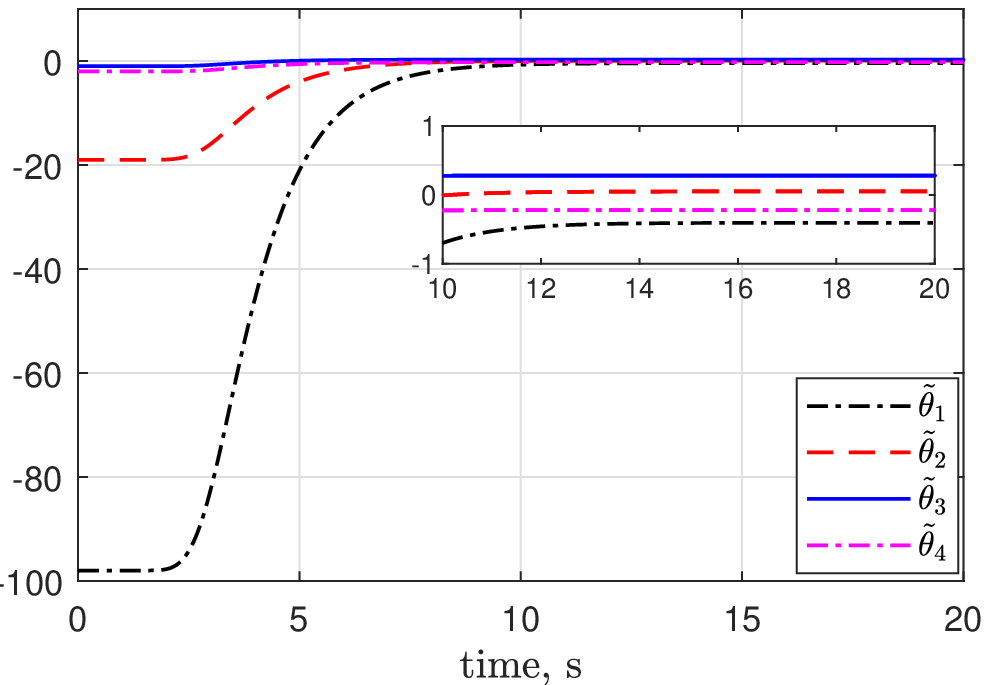}
     \caption{$\tilde{\theta}(t)$ (under the influence of $\eta(t)$) with DREM procedure and new filter, estimator \eqref{eq:gradAlgNew} and $u_{pb}(t)$}
     \label{fig:u2_thetaDREMIMP_noise}
 \end{figure}
%
\section{Simulations of the Discrete-time LRE Generator}
\label{sec6}
%
In this section, we present comparative simulations of the DT estimation of a scalar parameter $\theta \in \rea$ using the standard gradient descent adaptation with the original and the new regressor, that is, 
\begin{align}
\label{gradori}
\hat \theta_\bfori (k+1)=&  \hat \theta_\bfori (k)+\gamma \Delta(k) \left[ \mathcal{Y} (k) -\Delta (k)  \hat \theta_\bfori (k) \right] ,\;\gamma>0   \\
\hat \theta_\bfnew (k+1)=& \hat \theta_\bfnew (k)+\gamma \Phi_2(k) \left[ Y_2 (k) - \Phi_2 (k) \hat \theta_\bfnew (k) \right] ,\;\gamma>0  ,
\label{gradnew}
\end{align}
for three different signals $\Delta (k)$, namely:  
  \begin{eqnarray}
\begin{aligned}
    \Delta_a(k) & = e^{-3 k}\\
    \Delta_b(k) & =
    \left\{
    \begin{aligned}
         1 & &k \in[0,0.2] 
         \\
         0 & & k > \mbox{0.2}
    \end{aligned}
    \right. \\
    \Delta_c(k) & = {1\over 7 + k}. 
\end{aligned}
\label{taus}
\end{eqnarray}

Clearly, the three signals are not  PE and belong to ${\mathcal L}_2$. Hence, according with Proposition \ref{pro3} the estimator \eqref{gradori} will not converge. On the other hand, since they are IE, the estimator \eqref{gradnew} should guarantee convergence for small values of $T$. 

In the first simulations we consider the unknown parameter $\theta=5$ and select $\beta=\frac{3}{4}$, $\mu=0.4$, $b=0.1$ and $\gamma=0.1$. The initial conditions of the estimators are set as $\hat \theta_{\bfori}(0)=0$ and $\hat \theta_{\bfnew}(0)=0$. In the light of the key assumption \eqref{tsqu} we also check the effect of the size of the constant $T$, carrying out simulations using the values $T=0.01$, $T=0.1$ and $T=1.5$. 

The results of these simulations are given in Figs. \ref{figexp}-\ref{figfrac}, that confirm the predictions of  the theoretical analysis. We also notice that taking a large value for $T$ does not affect the steady-state performance, but it increases significantly the convergence time. The rationale for this behavior may be explained as follows. In the scenarios considered above the signals $\Delta(k)$ converge to zero, reducing the effect of the truncation error. 

\begin{figure}[htp]
    \centering
	\includegraphics[width = 0.8\textwidth]{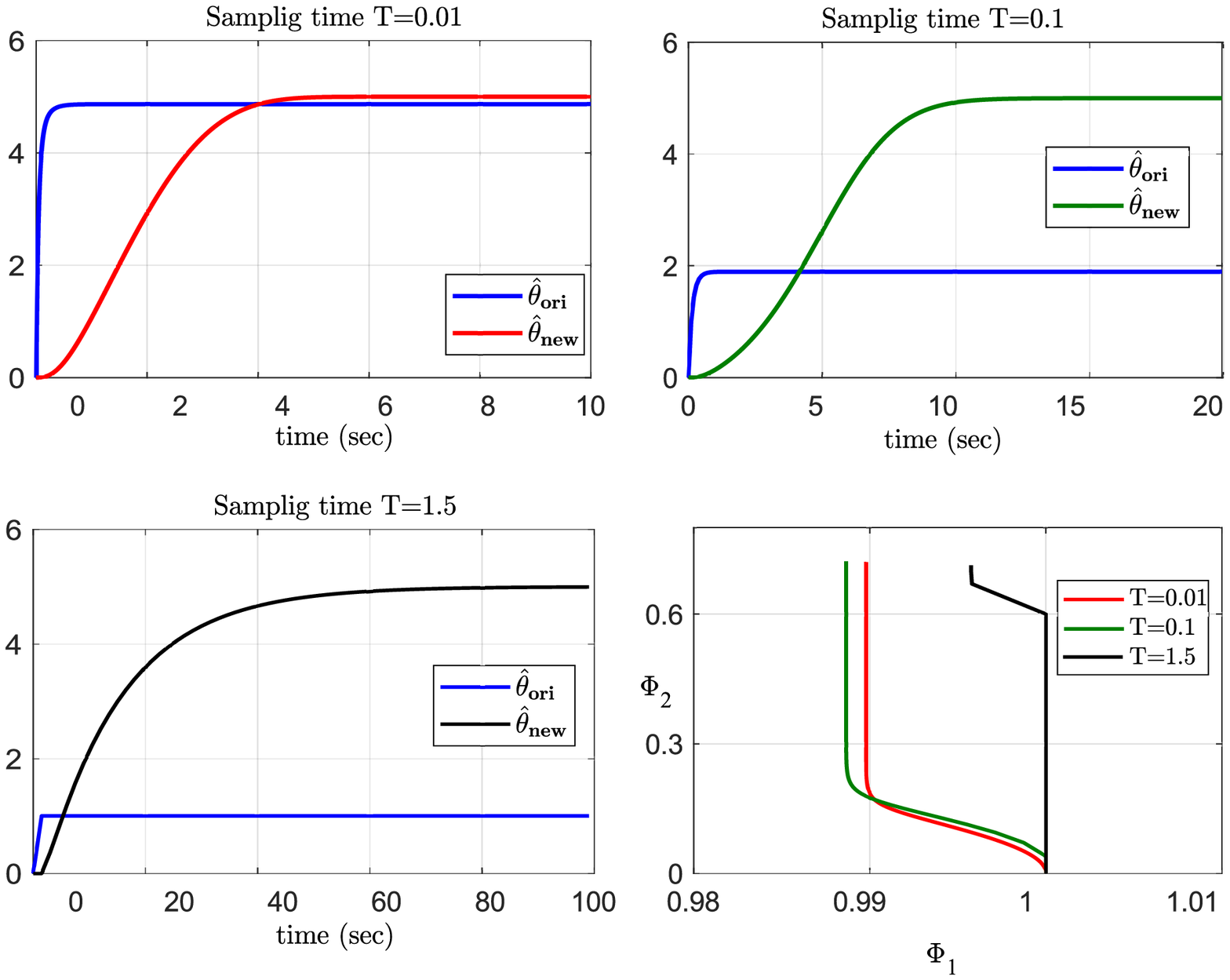}
    \caption{Estimates using the original and new  LRE and the phase portrait  of $\Phi(k)$ with the signal $\Delta_a(k)$ }
    \label{figexp}
\end{figure}

\begin{figure}[htp]
    \centering
	\includegraphics[width = 0.75\textwidth]{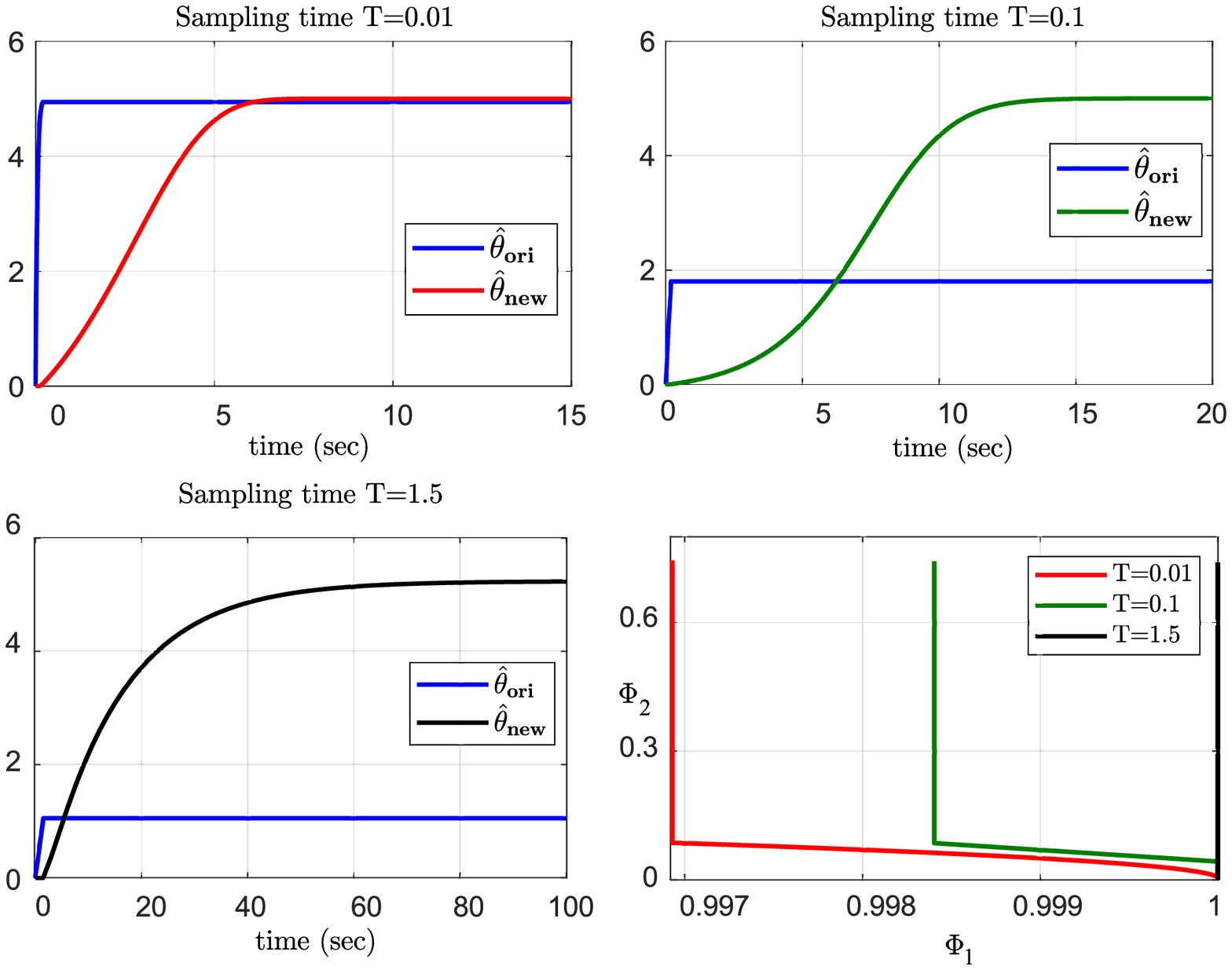}
    \caption{Estimates using the original and new  LRE and the phase portrait  of $\Phi(k)$ with the signal $\Delta_b(k)$  }
    \label{figesca}
\end{figure}

\begin{figure}[htp]
    \centering
	\includegraphics[width = 0.8\textwidth]{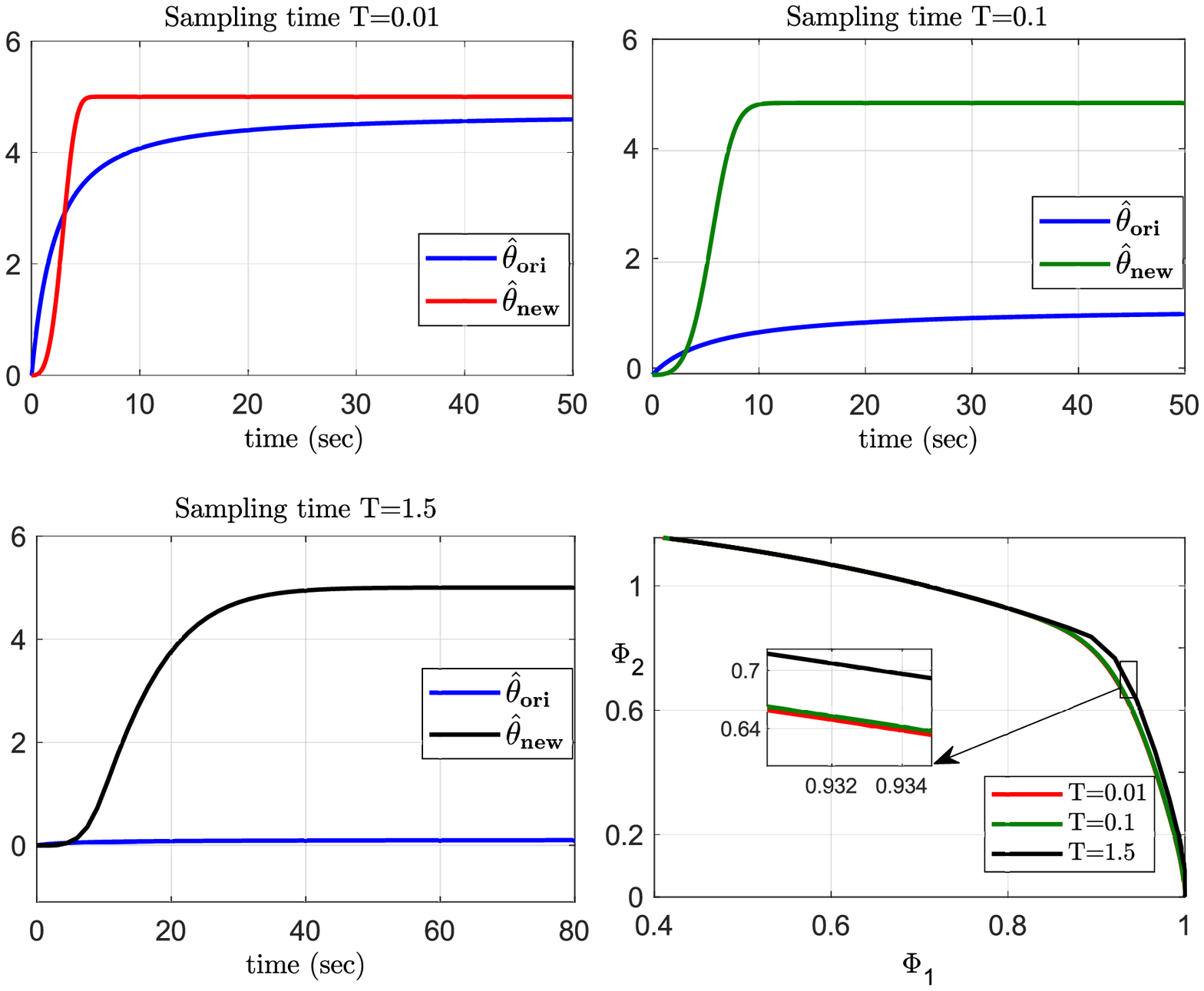}
    \caption{Estimates using the original and new  LRE and the phase portrait    of $\Phi(k)$ with the signal $\Delta_c(k)$ }
    \label{figfrac}
\end{figure}

The situation is different if $\Delta(k)$ {\em does not} converge to zero, for instance if it is PE. In that case, it is expected that the performance is degraded with increasing values of $T$. To validate this conjecture we carried  out a simulation with the PE signal $\Delta_d(k)=\cos\left(\frac{\pi}{4}k\right)$. In this case the estimator \eqref{gradori} always converges. However, \eqref{gradnew}  will ensure parameter convergence only for small values of $T$. This is corroborated with the plots of  Fig. \ref{figcos} that show how the performance of the estimator  \eqref{gradnew} degrades with increasing $T$. Moreover, from the   phase portrait we notice that for $T=1.5$ the signal $\Phi(k)$ does not live in the gray section indicated in Fig. \ref{fig1}, violating the predictions of the theory because  assumption \eqref{tsqu} is not valid anymore. Furthermore, using the same $\Delta_d(k)$ and $T=1.5$, and choosing a large adaptation gain $\gamma=1.6$---in contrast to $\gamma=0.1$ used before -- the  estimator  \eqref{gradnew} becomes {\em unstable} as shown in  Fig. \ref{figcos2}. It is important to underscore that the main motivation for the introduction of the LRE generator is for the case when the original regressor is not PE, therefore the scenario considered in these simulations will not be encountered in practice.

\begin{figure}[htp]
    \centering
	\includegraphics[width = 0.8\textwidth]{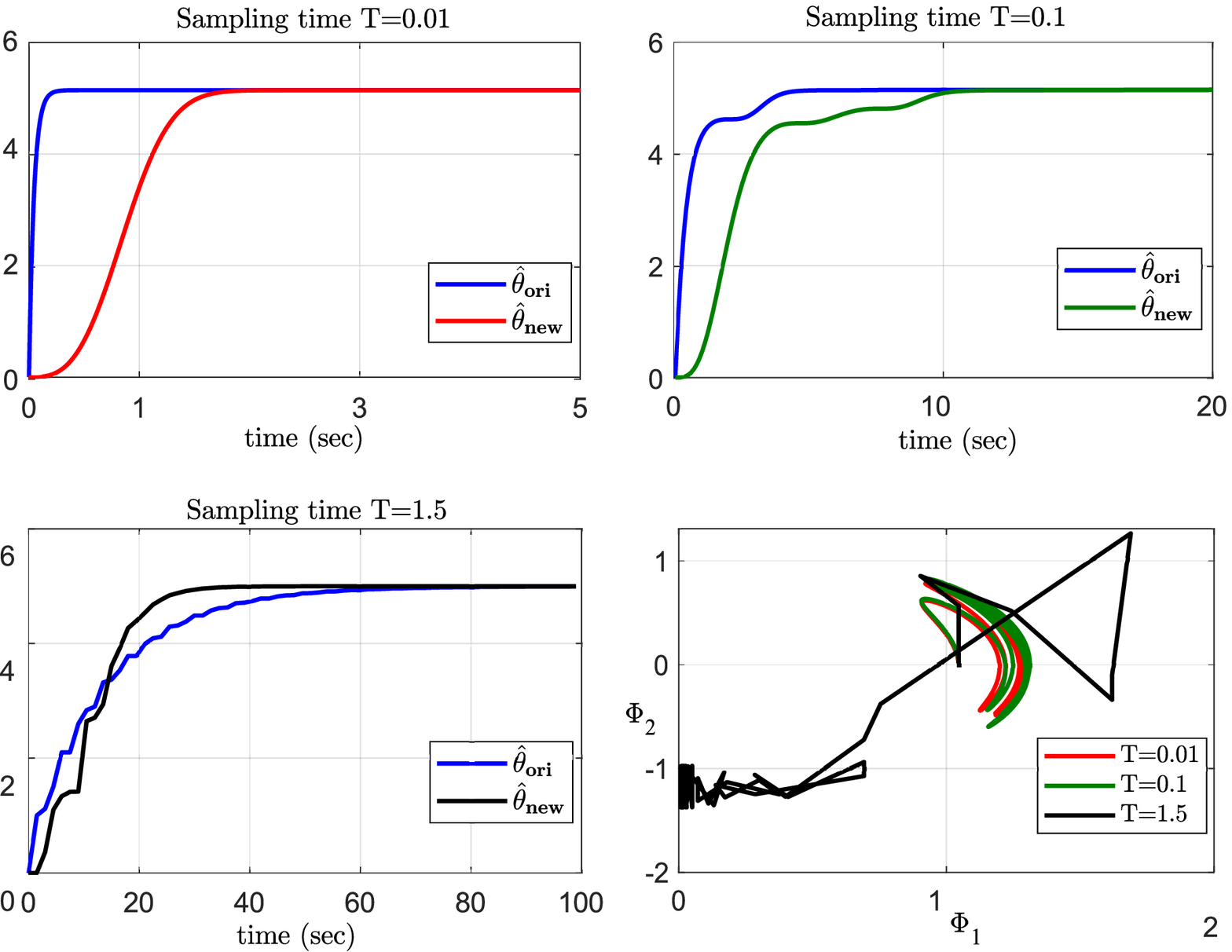}
    \caption{Estimates using the original and new  LRE and the phase portrait  of $\Phi(k)$ with the signal $\Delta_d(k)$ }
    \label{figcos}
\end{figure}

\begin{figure}[htp]
    \centering
	\includegraphics[width = 0.8\textwidth]{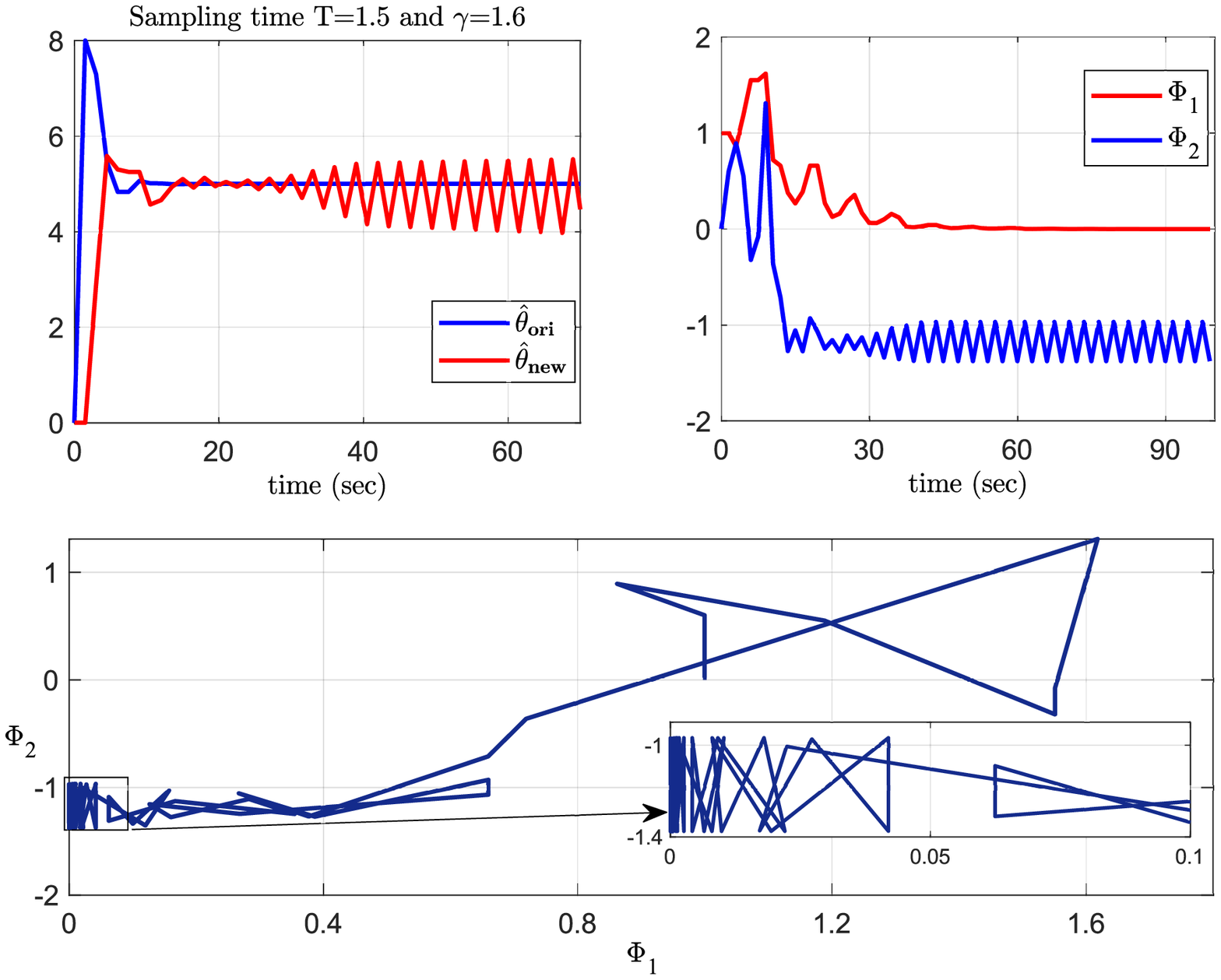}
    \caption{Estimates using the original and new  LRE and the phase portrait  of $\Phi(k)$ with the signal $\Delta_d(k)$ and $\gamma=1.6$}
    \label{figcos2}
\end{figure}

%
\section{Conclusions and Future Research}
\label{sec7}
%
We have proposed new CT and DT estimators for the LRE \eqref{orilre} that ensure global exponential convergence  under the weak assumption that the regressor $\Omega$ is IE. To the best of our knowledge, this is the first time that such a result is established for a truly on-line estimator.\footnote{As mentioned in Section \ref{sec1} the concurrent \cite{CHOetal} and composite learning \cite{PANYU} estimators involve an off-line operation of data monitoring and stacking. See also \cite{ORTproieee} where it is shown that allowing off-line calculations it is possible to ensure finite convergence time with an IE assumption.}

Our current efforts are directed towards the relaxation of the assumption  \eqref{tsqu} for the DT estimator and to further study the effect of noise in the estimators performance.
%

\appendix
\section{List of Acronyms}\medskip
%
\begin{table}[h]
	\centering
	\label{tab:2}
	\renewcommand\arraystretch{1.6}
	\begin{tabular}{l|r}
		\hline\hline
		CT & Continuous-time \\
		DREM  &  Dynamic regressor extension and mixing \\
		DT & Discrete-time \\
		IE & Interval excitation \\		
		KRE & Kreisselmeier's regressor extension \\
		LRE & Linear regressor equation \\
		LTI &  Linear time-invariant  \\
		PE &  Persistent excitation  \\
		\hline\hline
	\end{tabular}
\end{table}
\end{document}